\documentclass[11pt]{article}

\usepackage[T1]{fontenc}
\usepackage[utf8]{inputenc}
\usepackage{lmodern}
\usepackage[margin=1in]{geometry}
\usepackage{microtype}
\usepackage{amsmath,amssymb,amsthm}
\usepackage{xcolor}
\usepackage{algorithm}
\usepackage[noend]{algpseudocode}
\usepackage{hyperref}
\usepackage[nameinlink,noabbrev]{cleveref}
\hypersetup{
  colorlinks=true,
  linkcolor=blue!50!black,
  citecolor=blue!50!black,
  urlcolor=blue!50!black,
  pdftitle={Extended One-Liners for the Gamma, Poisson, and Binomial Distributions},
  pdfauthor={Dylan Greaves}
}

\theoremstyle{plain}
\newtheorem{theorem}{Theorem}
\newtheorem{lemma}[theorem]{Lemma}
\newtheorem{proposition}[theorem]{Proposition}
\theoremstyle{remark}
\newtheorem{remark}[theorem]{Remark}

\Crefname{theorem}{Theorem}{Theorems}
\Crefname{lemma}{Lemma}{Lemmas}
\Crefname{algorithm}{Algorithm}{Algorithms}
\makeatletter
\providecommand{\theHALG@line}{\thealgorithm.\arabic{ALG@line}}
\makeatother

\newcommand{\1}{\mathbf 1}
\newcommand{\R}{\mathbb R}
\newcommand{\E}{\mathbb E}
\newcommand{\Pp}{\mathbb P}
\newcommand{\Unif}{\operatorname{Unif}}
\newcommand{\Pois}{\operatorname{Poisson}}
\newcommand{\Bin}{\operatorname{Binomial}}
\newcommand{\Gam}{\operatorname{Gamma}}
\newcommand{\Cat}{\operatorname{Categorical}}
\newcommand{\dd}{\,d}
\newcommand{\eqdist}{\mathrel{\overset{\mathrm{d}}{=}}}

\newcommand{\paperTitle}{Extended One-Liners for the Gamma, Poisson, and Binomial Distributions}
\title{\paperTitle}
\author{Dylan Greaves\\\texttt{dgreaves103@gmail.com}}
\date{}

\begin{document}
\maketitle

\begin{abstract}
  We demonstrate explicit transformations of three independent uniforms with exact Gamma$(a,1)$, $a>0$,
  Poisson$(\lambda)$, $\lambda>0$, and Binomial$(N,p)$, $N\geq 1$, $0\leq p \leq 1$ laws using only
  elementary operations.
\end{abstract}

\section{Introduction}
\label{sec:intro}

When possible, representing a family of distributions
$\{P_\theta\}_{\theta\in\Theta}$ as a simple transformation $T$ of $k$ i.i.d.\ uniform random variates
\[
  T(U_1,\ldots,U_k;\theta)\sim P_\theta,
  \qquad
  U_1,\ldots,U_k\stackrel{\mathrm{i.i.d.}}{\sim}\operatorname{Unif}(0,1),
\]
has several appealing properties. These methods are
easy to implement, consume a fixed number of uniforms,
which is useful in common-random-number couplings, and their
lack of loops can make them competitive as sampling methods on GPUs~\cite{howes2007cuda}.
Additionally, if $T$ is sufficiently
regular, these transformations find application in quasi-Monte Carlo~\cite{practicalqmc} and machine
learning~\cite{figurnov2018implicit,mohamed2020montecarlo}.

To define the class of simple transformations, Devroye~\cite{Devroye1996} introduces the notion of a
one-liner. We say a transformation is a
\emph{one-liner} if it can be written as an expression tree whose leaves are i.i.d.\ uniform random variates, parameters $\theta$, and constants, and whose internal nodes are operators belonging to a permitted class $\mathcal F$ of elementary operations, e.g.
\[
  \mathcal{F}=\{+,-,\times,/,\bmod,\operatorname{round},|\cdot|,
  \operatorname{sign},\lfloor\cdot\rfloor,\lceil\cdot\rceil,
  \sin,\cos,\exp,\log,\tan,\arctan\}.
\]
While $\mathcal F$ is fairly restrictive, we can express several commonly used operations by
composing these primitive operations, for example branching and powers:
\[
  \1_{\{x>y\}}=\lceil\operatorname{sign}(x-y)/2\rceil,
  \qquad
  x^y=\exp(y\log x),\quad x>0.
\]
If intermediate nodes can be reused, so that the transformation can be represented as an expression
DAG rather than a tree, we call the transformation an \emph{extended
one-liner}.
We call the
family $\{P_\theta\}_{\theta\in\Theta}$ \emph{$k$-simple}~\cite{Devroye2006} if there is an
extended one-liner that generates
$P_\theta$ for every $\theta\in\Theta$ using at most $k$ independent
uniforms.

Several well-known families are known to be $k$-simple for some $k$, including the normal
distribution for $k=2$ by the Box--Muller transformation,
Student's t distribution, and the symmetric beta distribution. Devroye singled out the general gamma
and Poisson distributions as two notable exceptions for which no extended one-liner was
known~\cite{Devroye1996,Devroye2006}.
Indeed, all widely used exact methods for these distributions rely on
either rejection sampling or
inversion; see, e.g.,
\cite{Devroye1986,MarsagliaTsang2000,Hormann1993,Knuth1997}.

A recent paper by the author showed that gamma is 3-simple for the subfamily with shape parameter
below one~\cite{GreavesBelowOne}.
In this paper, we answer Devroye's original question in the affirmative: the general gamma and
Poisson are 3-simple. Additionally, binomial is shown to be 3-simple using a similar technique.
In all three cases, the core idea is to use a generalization of the acceptance-complement method
with a randomized acceptance threshold. To avoid computing normalizing constants not directly
available as operations in $\mathcal F$,
we construct non-negative bounded random variables that equal these constants in expectation, taking
advantage of some contour integral identities.

Throughout, we use the parameterizations:
\[
  \begin{aligned}
    \Gam(a,\theta):\quad &f(x;a,\theta)=\frac{x^{a-1}e^{-x/\theta}}{\Gamma(a)\theta^a},
    &&x>0,\quad a,\theta>0,\\
    \Pois(\lambda):\quad &p(k;\lambda)=e^{-\lambda}\frac{\lambda^k}{k!},
    &&k=0,1,\ldots,\quad\lambda>0,\\
    \Bin(N,p):\quad &p(k;N,p)=\binom Nk p^k(1-p)^{N-k},
    &&k=0,\ldots,N.
  \end{aligned}
\]
We suppress the scale parameter $\theta$, using $\Gam(a)$ to denote $\Gam(a,1)$, as any one-liner
for $\Gam(a)$ extends to $\Gam(a,\theta)$ by the identity $\Gam(a,\theta)\eqdist \theta\Gam(a)$.
\section{Preliminaries}
\label{sec:preliminaries}

\subsection{Generalized Acceptance-Complement}
\label{sec:acceptance-complement}
All three extended one-liners use a generalization of the acceptance-complement
method~\cite{KronmalPeterson1981,KronmalPeterson1984} with a randomized acceptance threshold,
analogous
to the generalized rejection method~\cite[Chapter~II, Theorem~3.4]{Devroye1986}, which we outline
below.
\begin{lemma}[Generalized acceptance-complement]\label{lem:residual}
  Let $\mu$ be a target distribution on a space $S$, and let $(Y,A)$ be
  an $S\times[0,1]$-valued random variable, where $Y$ is the proposed value and $A$
  its acceptance probability.  Define the acceptance subprobability measure
  \[
    Q(B)=\E[A\1_{\{Y\in B\}}],
  \]
  and suppose that $\E A<1$ and $Q(B)\leq\mu(B)$ for every measurable set $B$.  Given
  $U\sim\Unif(0,1)$ independent of $(Y,A)$, return $Y$ if $U\leq A$;
  otherwise return an independent draw from the distribution
  \[
    R(B)=\frac{\mu(B)-Q(B)}{1-\E A}.
  \]
  Then the returned random variable has distribution $\mu$.
\end{lemma}

\begin{proof}
  Let $\widehat X$ denote the value returned by the algorithm.  For every measurable set
  $B$, we have
  \[
    \begin{aligned}
      \Pp(\widehat X\in B)
      &=\Pp(U\leq A,\,Y\in B)+\Pp(U>A)R(B)\\
      &=\E[A\1_{\{Y\in B\}}]+(1-\E A)R(B)\\
      &=Q(B)+(1-\E A)R(B)\\
      &=\mu(B).
    \end{aligned}
  \]
\end{proof}
In particular, when the target and proposal distributions have densities,
Lemma~\ref{lem:residual} gives the following:

\begin{lemma}\label{lem:conditional-repair}
  Let $p$ be a target density, let the proposal $Y$ have density $q$, and let
  $A\in[0,1]$ be a random acceptance probability, possibly dependent on $Y$.
  Suppose $r\geq0$, $\int r>0$, and $0\leq p-r\leq q$.  Suppose also that
  \begin{equation}
    \E[A\mid Y=y]=\frac{p(y)-r(y)}{q(y)},
    \qquad q(y)>0.
    \label{eq:conditional-repair}
  \end{equation}
  Generate a pair $(Y, A)$, and return $Y$ with probability
  $A$; otherwise return an independent random variate with density $r/\int r$. Then the output is a random variate
  with density $p$.
\end{lemma}

Thus, to construct an extended one-liner for $p$, it suffices to 1) find $q$ and $r$ that are easy
to sample from satisfying the above inequalities, and 2) find a one-liner $A$ with conditional
expectation equal to
\eqref{eq:conditional-repair}.

\subsection{Random Weights}

\label{sec:bounded-contour-weights}

The main challenge of realizing $A$ as a one-liner is that the target density
involves terms $1/\Gamma(a)$, $1/k!$, and $\binom{N}{k}$ that
are not directly available as operations in $\mathcal{F}$. To avoid evaluating these, we
construct bounded nonnegative random variables,
expressible as one-liners, that equal these constants in expectation.
\begin{lemma}[Reciprocal-gamma weight]\label{lem:reciprocal-gamma-weight}
  For $0\leq\theta<\pi$, put
  \[
    \Phi(\theta)=1-\theta\cot\theta+\log\frac{\theta}{\sin\theta},
    \qquad \Phi(0)=0,
  \]
  and, for $a>0$, define
  \begin{equation}
    W_a(\theta)=
    \begin{cases}
      \dfrac{(e/a)^a}{\sqrt{2\pi a}}
      \exp\left\{-a\left(\Phi(\theta)-\dfrac{\theta^2}{2}\right)\right\},
      &0\leq\theta<\pi,\\[2mm]
      0,&\theta\geq\pi.
    \end{cases}
    \label{eq:reciprocal-weight}
  \end{equation}

  For $a>0$, if $Z\sim N(0,1)$ and
  \[
    \Theta=\frac{|Z|}{\sqrt a},
  \]
  then
  \begin{equation}
    \E W_a(\Theta)=\frac1{\Gamma(a+1)},
    \qquad
    0\leq W_a\leq\frac{(e/a)^a}{\sqrt{2\pi a}}. \label{eq:reciprocal-facts}
  \end{equation}
\end{lemma}

\begin{proof}
  See Appendix~\ref{app:reciprocal-gamma-weight}.
\end{proof}

\begin{lemma}[Binomial weight]\label{lem:binomial-weight}
  For integers $0<k<N$, set
  \[
    x=\frac{k}{N},\qquad y=1-x,
  \]
  and, for $0<\theta<\pi$, put
  \[
    \Phi_x(\theta)
    =\log\frac{\theta}{\sin\theta}
    -x\log\frac{x\theta}{\sin(x\theta)}
    -y\log\frac{y\theta}{\sin(y\theta)},
    \qquad \Phi_x(0)=0,
  \]
  and
  \begin{equation}
    W_{N,k}(\theta)=
    \begin{cases}
      \dfrac1{\sqrt{2\pi Nxy}}
      \exp\left\{-N\left(\Phi_x(\theta)-\dfrac{xy\theta^2}{2}\right)\right\},
      &0\leq\theta<\pi,\\[2mm]
      0,&\theta\geq\pi.
    \end{cases}
    \label{eq:binomial-weight}
  \end{equation}

  If $Z\sim N(0,1)$ and
  \[
    \Theta=\frac{|Z|}{\sqrt{Nxy}},
  \]
  then
  \begin{equation}
    \E W_{N,k}(\Theta)=\binom Nkx^ky^{N-k},
    \qquad 0\leq W_{N,k}\leq\frac1{\sqrt{2\pi Nxy}}. \label{eq:binomial-weight-facts}
  \end{equation}
\end{lemma}

\begin{proof}
  See Appendix~\ref{app:binomial-weight}.
\end{proof}

\section{Gamma Extended One-Liner}
\label{sec:gamma}

We now construct an extended one-liner for $\Gam(a)$ with $a>0$.  It
suffices to find a one-liner for $a\geq1$, since any such method automatically
extends to $0<a<1$ using the beta--gamma identity
\[
  U^{1/a}\Gam(a+1)\eqdist\Gam(a),
  \qquad U\sim\Unif(0,1).
\]
Alternatively, one can use the one-liner for $\Gam(a)$, $a<1$, developed in
\cite{GreavesBelowOne}.

Suppose now that $a\geq1$.  Our strategy will be to modify the
Marsaglia--Tsang rejection method \cite{MarsagliaTsang2000} into an extended
one-liner using the acceptance-complement method in
Lemma~\ref{lem:residual}.  Following Marsaglia and Tsang, set
\[
  d=a-\frac13,\qquad c=\frac{1}{3\sqrt d},\qquad
  h(x)=d(1+cx)^3.
\]
Marsaglia and Tsang note that if one can generate a random variate $X$ with
density $g_a/\Gamma(a)$, where
\begin{equation}
  \begin{aligned}
    g_a(x)
    &=h'(x)h(x)^{a-1}e^{-h(x)}\\
    &=d^{a-1/2}(1+cx)^{3d}
    e^{-d(1+cx)^3},\qquad x>-1/c,
  \end{aligned} \label{eq:gamma-pullback}
\end{equation}
and 0 otherwise, then $h(X)$ has the
desired $\Gam(a)$ density by change of variables; hence it suffices to find a one-liner for
$g_a/\Gamma(a)$.
To that end, let
\[
  U_1,U_2,U_3\overset{\mathrm{i.i.d.}}{\sim}\Unif(0,1)
\]
and construct the independent standard-normal pair
\[
  (Z_1,Z_2)=\sqrt{-2\log U_1}
  \bigl(\cos(2\pi U_2),\sin(2\pi U_2)\bigr)
\]
using the Box--Muller transform.  Define
\[
  \widehat w_a(z)=aW_a(|z|/\sqrt a),
\]
where
\begin{equation}
  \E\widehat w_a(Z_1)=\frac1{\Gamma(a)}, \label{eq:gamma-weight-mean}
\end{equation}
by Lemma~\ref{lem:reciprocal-gamma-weight}.
Define
\[
  \delta_a=\frac13+\left(a-\frac12\right)\log\frac da,
  \qquad \eta_a=1-e^{-\delta_a},
\]
and let
\begin{equation}
  r_a(x)=\frac65\eta_ag_a(0)
  \left(1-\frac{|x|}{2}\right)_+. \label{eq:gamma-repair}
\end{equation}
We use the proposal $Y=Z_2$ with acceptance probability
\[
  A=\widehat w_a(Z_1)
  \frac{g_a(Z_2)-r_a(Z_2)}{\phi(Z_2)}.
\]
where $\phi(x) = \frac{1}{\sqrt{2\pi}}e^{-x^2/2}$ is the standard-normal density.
The following lemma shows that $0\leq A\leq 1$:
\begin{lemma}[Gamma envelope]\label{lem:gamma-envelope}
  For $a\geq1$ and $x,z\in\R$,
  \[
    0\leq r_a(x)\leq g_a(x),
    \qquad
    0\leq\widehat w_a(z)\{g_a(x)-r_a(x)\}\leq\phi(x).
  \]
\end{lemma}

\begin{proof}
  See Appendix~\ref{app:gamma-envelope}.
\end{proof}

By independence and \eqref{eq:gamma-weight-mean}, for $z\in\R$,
\begin{equation}
  \E[A\mid Z_2=z]
  =\frac{g_a(z)-r_a(z)}{\Gamma(a)\phi(z)}
  =\frac{g_a(z)/\Gamma(a)-r_a(z)/\Gamma(a)}{\phi(z)}.
  \label{eq:gamma-conditional-mean}
\end{equation}
By Lemma~\ref{lem:conditional-repair} with
$p=g_a/\Gamma(a)$ and $q=\phi$, the required complement arm density is the triangular density
proportional to $r_a$.  If $V\sim\Unif(0,1)$ independently, we can draw a random variate from this
distribution
using the transformation $T(V)$, where
\begin{equation}
  T(v)=
  \begin{cases}
    2(\sqrt{2v}-1),&0<v\leq1/2,\\
    2\{1-\sqrt{2(1-v)}\},&1/2<v<1.
  \end{cases} \label{eq:triangle-inverse}
\end{equation}
While not necessarily recommended in practice, one can avoid drawing an additional uniform $V$ here
by uniform recycling: on the complement branch $U_3>A$, if
$V\gets(U_3-A)/(1-A)$, then
$V\mid U_3>A\sim\Unif(0,1)$; see
\cite[Section~3.7]{Devroye1986}.

\begin{algorithm}[H]
  \caption{\textsc{GammaCore}$(a)$: $\Gam(a)$ sampler for $a\geq1$}
  \label{alg:gamma-core}
  \begin{algorithmic}[1]
    \Require $a\geq1$
    \State $U_1,U_2,U_3\overset{\mathrm{i.i.d.}}{\sim}\Unif(0,1)$
    \State $R\gets\sqrt{-2\log U_1}$, $Z_1\gets R\cos(2\pi U_2)$, $Z_2\gets R\sin(2\pi U_2)$
    \State $d\gets a-1/3$, $c\gets1/(3\sqrt d)$, $A\gets0$
    \State \textbf{if} $1+cZ_2>0$ \textbf{then}
    \State \hspace{1.2em}$A\gets\widehat w_a(Z_1)\{g_a(Z_2)-r_a(Z_2)\}/\phi(Z_2)$
    \Comment{Alternatively, $1/\Gamma(a)$ for $\widehat w_a(Z_1)$ if $\Gamma$ is permitted}
    \State \textbf{if} $A>0$ and $U_3<A$ \textbf{then}
    \State \hspace{1.2em}\textbf{return} $d(1+cZ_2)^3$
    \State $V\gets(U_3-A)/(1-A)$, $X\gets T(V)$
    \State \textbf{return} $d(1+cX)^3$
  \end{algorithmic}
\end{algorithm}

\begin{algorithm}[H]
  \caption{$\Gam(a)$ sampler}
  \label{alg:gamma}
  \begin{algorithmic}[1]
    \Require $a>0$
    \State \textbf{if} $a\geq1$ \textbf{then}
    \State \hspace{1.2em}\textbf{return} $\textsc{GammaCore}(a)$
    \State \textbf{return} $\textsc{GammaBelowOne}(a)$
    \Comment{Alternatively, $U_0^{1/a}\textsc{GammaCore}(a+1)$}
  \end{algorithmic}
\end{algorithm}

\begin{theorem}[Gamma]\label{thm:gamma}
  \Cref{alg:gamma-core} returns $\Gam(a)$ for every $a\geq1$, and
  \Cref{alg:gamma} does so for every $a>0$.
  Thus the $\Gam(a)$ law is 3-simple.
\end{theorem}

\begin{proof}
  Assume $a\geq1$.  In Lemma~\ref{lem:conditional-repair}, take $Y=Z_2$ and
  \[
    p(x)=\frac{g_a(x)}{\Gamma(a)},\qquad q(x)=\phi(x),\qquad
    r(x)=\frac{r_a(x)}{\Gamma(a)}.
  \]
  Lemma~\ref{lem:gamma-envelope} gives $0\leq A\leq1$, while
  \eqref{eq:gamma-conditional-mean} gives
  $\E[A\mid Y=x]=\{p(x)-r(x)\}/q(x)$.  On the complement
  branch, $T(V)$ has density proportional to $r$.  Hence
  $X\sim g_a/\Gamma(a)$, and $h(X)\sim\Gam(a)$ by \eqref{eq:gamma-pullback}.
  The claim follows by using the three-uniform gamma one-liner
  for $0<a<1$ from \cite{GreavesBelowOne}
  (\Cref{alg:gamma-below-one} in Appendix~\ref{app:gamma-below-one}) on the $a<1$ branch.
\end{proof}

\begin{remark}
  If evaluation of the gamma function is permitted, we can replace $\widehat w_a(Z_1)$ by
  $1/\Gamma(a)$ directly, still satisfying the assumptions of
  Lemma~\ref{lem:conditional-repair}.
\end{remark}

\section{Discrete Extended One-Liners}
\label{sec:discrete}

Both discrete one-liners use a special case of
Lemma~\ref{lem:residual}, where the acceptance subprobability $Q$ is set to equal the target
probability
mass function $(\pi_k)_{k\in\mathbb Z}$ exactly outside an excluded set of atoms $S$, and is 0
otherwise:
\begin{proposition}[Discrete repair]\label{prop:finite-gaussian-repair}
  Let $(\pi_k)_{k\in\mathbb Z}$ be the target probability mass function, and let
  $S\subseteq\mathbb Z$ with $\sum_{k\in S}\pi_k>0$.  Let
  $Z\sim N(0,1)$, let $T:\mathbb R\to\mathbb R$ be increasing, and set
  $Y=\lfloor T(Z)\rfloor$.  For every $k\notin S$ with $\pi_k>0$, define the
  interval $I_k=T^{-1}([k,k+1))$ and suppose $0<|I_k|<\infty$.  For each such
  $k$, let $\widehat\pi_k$ be a
  nonnegative random variable, independent of $Z$, satisfying
  \begin{equation}\label{eq:discrete-repair-pi-bound}
    \E\widehat\pi_k=\pi_k,
    \qquad
    \widehat\pi_k\leq |I_k|\inf_{w\in I_k}\phi(w)
    \quad\text{a.s.}
  \end{equation}
  Define
  \[
    A=
    \begin{cases}
      \dfrac{\widehat\pi_Y}{|I_Y|\phi(Z)},
      &Y\notin S,\ \pi_Y>0,\ |I_Y|>0,\\[2mm]
      0,&\text{otherwise}.
    \end{cases}
  \]
  Then $0\leq A\leq1$.  Return $Y$ with probability $A$; otherwise return an
  independent random variate with probability mass function proportional to
  $(\pi_k)_{k\in S}$.  Then the output is a random variate with probability mass
  function $(\pi_k)_{k\in\mathbb Z}$.
\end{proposition}

\begin{proof}
  The assumptions show that $0\leq A\leq1$.  For $k\notin S$ with $\pi_k>0$,
  by independence
  \[
    \begin{aligned}
      Q(\{k\})
      &=\E[A\1_{\{Y=k\}}]\\
      &=\int_{I_k}\phi(z)
      \frac{\E\widehat\pi_k}{|I_k|\phi(z)}\dd z
      =\pi_k.
    \end{aligned}
  \]
  For $k\in S$ or $\pi_k=0$, $Q(\{k\})=0$.  Thus
  $Q(\{k\})=\pi_k\1_{\{k\notin S\}}$. The complement
  arm distribution is
  \[
    R(\{k\})=
    \begin{cases}
      \displaystyle\frac{\pi_k}{\sum_{j\in S}\pi_j},&k\in S,\\[2mm]
      0,&k\notin S,
    \end{cases}
  \]
  hence by Lemma~\ref{lem:residual}, the output has the claimed distribution.

\end{proof}

In the constructions below, $S$  has at most three atoms for Poisson and five for binomial,
so the complement arm distribution $R$ can be sampled from by inversion as a
one-liner: for a probability vector $\mathbf{p}=(p_0,\dots,p_m)$ and nonnegative weights
$\mathbf{w}=(w_0,\ldots,w_m)$ with $W=\sum_{j=0}^m w_j>0$ and $w_j\propto p_j$, define
\begin{equation}
  \Cat(v;\mathbf{w})
  =\sum_{j=0}^{m-1}
  \1_{\left\{v\geq (w_0+\cdots+w_j)/W\right\}},
  \label{eq:categorical-inversion}
\end{equation}
then for $V\sim\Unif(0,1)$, we have
\[
  \Pp\{\Cat(V;\mathbf{w})=j\}=p_j,
  \qquad j=0,\ldots,m.
\]
As $w_j$ need only be proportional to $p_j$, we can avoid computing $p_j$ directly;
for Poisson and binomial, we take advantage of recurrences these probabilities satisfy.

The proposal transformations used below are motivated by the fact that, away from endpoints,
\[
  \widehat\pi_k\leq M(k):=
  \frac{e^{-D(k)}}{\sqrt{2\pi V(k)}},
  \qquad V(k)>0,
\]
where
\[
  D(t)=
  \begin{cases}
    D_{\mathrm P}(t\Vert\lambda),&\text{Poisson},\\
    N D_{\mathrm B}(t/N\Vert p),&\text{binomial},
  \end{cases}
  \qquad
  V(t)=
  \begin{cases}
    t,&\text{Poisson},\\
    t(N-t)/N,&\text{binomial},
  \end{cases}
\]
and $D_{\mathrm P}$ and $D_{\mathrm B}$ are the Poisson and Bernoulli
Kullback--Leibler divergences
\[
  D_{\mathrm P}(t\Vert\lambda)
  =t\log\frac{t}{\lambda}-t+\lambda,
  \qquad
  D_{\mathrm B}(x\Vert p)
  =x\log\frac{x}{p}+(1-x)\log\frac{1-x}{1-p},
\]
with $0\log0=0$.
In order to show that
\begin{equation}\label{eq:discrete-envelope-bound}
  \widehat\pi_k\leq M(k)
  \leq |I_k|\inf_{z\in I_k}\phi(z),\qquad k\notin S,
\end{equation}
we choose an increasing coordinate function $\rho$, with elementary inverse, satisfying
\[
  D(t)\geq\frac12\rho(t)^2,
\]
so that
\[
  \phi(\rho(t))\geq\frac{e^{-D(t)}}{\sqrt{2\pi}},\quad
\]
allowing us to bound the density term in \eqref{eq:discrete-envelope-bound} for proposals based on
$\rho^{-1}$. After an appropriate piecewise constant offset to $\rho^{-1}$, the inequality
\eqref{eq:discrete-envelope-bound} holds, which we prove in
Appendix~\ref{app:envelope-proofs}.
\subsection{Poisson}
\label{sec:poisson}

Fix $\lambda>0$.  Let
\[
  U_1,U_2,U_3\overset{\mathrm{i.i.d.}}{\sim}\Unif(0,1)
\]
and let $(Z_1,Z_2)$ be the corresponding Box--Muller pair.  Set
$\widehat\pi_0(z)=e^{-\lambda}$.  For $k\geq1$, set
$\theta=|z|/\sqrt k$ and define
\begin{equation}
  \widehat\pi_k(z)
  =e^{-\lambda}\lambda^k W_k(\theta) =
  \begin{cases}
    \dfrac{\exp\left[-D_{\mathrm P}(k\Vert\lambda)
    -k\left\{\Phi(\theta)-\theta^2/2\right\}\right]}{\sqrt{2\pi k}},
    &\theta<\pi,\\[2mm]
    0,&\theta\geq\pi.
  \end{cases}
  \label{eq:poisson-random-mass}
\end{equation}
By Lemma~\ref{lem:reciprocal-gamma-weight}, we have
\begin{equation}
  \E\widehat\pi_k(Z_2)
  =e^{-\lambda}\frac{\lambda^k}{k!}
  =\Pp\{\Pois(\lambda)=k\}.
  \label{eq:poisson-unbiasedness}
\end{equation}
Define the coordinate $\rho$ and its inverse by
\begin{align}
  \rho(t)
  &=\begin{cases}
    (t-\lambda)/\sqrt\lambda,&t\leq\lambda,\\[1mm]
    2(\sqrt t-\sqrt\lambda),&t\geq\lambda,
  \end{cases} \label{eq:poisson-rho}\\
  \rho^{-1}(z)
  &=\begin{cases}
    \lambda+\sqrt\lambda z,&z<0,\\[1mm]
    (\sqrt\lambda+z/2)^2,&z\geq0.
  \end{cases} \label{eq:poisson-rho-inverse}
\end{align}
Let
\[
  s=\max\{0,\lfloor\lambda\rfloor-1\},
  \qquad S=\{s,s+1,s+2\},
\]
and define
\[
  T(z)=\begin{cases}
    \rho^{-1}(z)-1,&z<0,\\
    \rho^{-1}(z)+1,&z\geq0,
  \end{cases}
  \qquad Y=\lfloor T(Z_1)\rfloor.
\]
For $k\geq0$ with $k\notin S$, set $I_k=T^{-1}([k,k+1))$:
\begin{equation}
  I_k=
  \begin{cases}
    [\rho(k+1),\rho(k+2)),&k+2\leq\lambda,\\[1mm]
    [\rho(k-1),\rho(k)),&k-1\geq\lambda.
  \end{cases} \label{eq:poisson-cells}
\end{equation}

Define
\[
  A=
  \begin{cases}
    \displaystyle
    \frac{\widehat\pi_Y(Z_2)}{|I_Y|\phi(Z_1)},
    &Y\geq0,\ Y\notin S,\\[2mm]
    0,&\text{otherwise}.
  \end{cases}
\]
\begin{lemma}[Poisson envelope]\label{lem:poisson-envelope}
  For $k\geq0$ with $k\notin S$, and for $z\in\R$,
  \begin{equation}\label{eq:poisson-envelope}
    \widehat\pi_k(z)
    \leq |I_k|\inf_{w\in I_k}\phi(w).
  \end{equation}
\end{lemma}

\begin{proof}
  See Appendix~\ref{app:poisson-envelope}.
\end{proof}

\begin{algorithm}[H]
  \caption{$\Pois(\lambda)$ sampler}
  \label{alg:poisson}
  \begin{algorithmic}[1]
    \Require $\lambda>0$
    \State $U_1,U_2,U_3\overset{\mathrm{i.i.d.}}{\sim}\Unif(0,1)$
    \State $R\gets\sqrt{-2\log U_1}$, $Z_1\gets R\cos(2\pi U_2)$, $Z_2\gets R\sin(2\pi U_2)$
    \State $s\gets\max\{0,\lfloor\lambda\rfloor-1\}$, $t\gets s+2$
    \State $K\gets\lfloor\rho^{-1}(Z_1)-1\rfloor$ \textbf{if} $Z_1<0$, \textbf{else} $K\gets\lfloor\rho^{-1}(Z_1)+1\rfloor$
    \State $A\gets0$
    \State \textbf{if} $K\geq0$ and $(K<s$ or $K>t)$ \textbf{then}
    \State \hspace{1.2em}$L\gets1/\sqrt\lambda$ \textbf{if} $K+2\leq\lambda$, \textbf{else} $L\gets2/(\sqrt K+\sqrt{K-1})$
    \State \hspace{1.2em}$A\gets\widehat\pi_K(Z_2)/(L\phi(Z_1))$
    \State \textbf{if} $A>0$ and $U_3<A$ \textbf{then}
    \State \hspace{1.2em}\textbf{return} $K$
    \State $V\gets(U_3-A)/(1-A)$
    \State $w_0\gets1$, $w_1\gets\lambda/(s+1)$, $w_2\gets\lambda^2/\{(s+1)(s+2)\}$
    \State \textbf{return} $s+\Cat(V;w_0,w_1,w_2)$
  \end{algorithmic}
\end{algorithm}

\begin{theorem}[Poisson]\label{thm:poisson}
  For every $\lambda>0$, \Cref{alg:poisson} returns a $\Pois(\lambda)$ variate.
  Thus the $\Pois(\lambda)$ law is 3-simple.
\end{theorem}

\begin{proof}
  Let $\pi_k=\Pp\{\Pois(\lambda)=k\}$.  Apply
  Proposition~\ref{prop:finite-gaussian-repair}, where the transformation $T$
  satisfies the required hypotheses, $\widehat\pi_k(Z_2)$ is unbiased by
  \eqref{eq:poisson-unbiasedness}, and the inequality
  \eqref{eq:discrete-repair-pi-bound} is given by
  \eqref{eq:poisson-envelope}.  Finally,
  $w_0,w_1,w_2$ are proportional to $\pi_s,\pi_{s+1},\pi_{s+2}$, so the complement arm has the required distribution.
\end{proof}

\subsection{Binomial}
\label{sec:binomial}

For binomial, fix an integer $N\geq1$ and $0<p<1$, and write $q=1-p$.  Let
\[
  U_1,U_2,U_3\overset{\mathrm{i.i.d.}}{\sim}\Unif(0,1)
\]
and let $(Z_1,Z_2)$ be the corresponding Box--Muller pair.  For $0<k<N$,
write $x=k/N$
and $y=1-x$.  Define
\begin{equation}
  \widehat\pi_k(z)=
  \begin{cases}
    q^N,&k=0,\\
    e^{-N D_{\mathrm B}(k/N\Vert p)}W_{N,k}(|z|/\sqrt{Nxy}),&0<k<N,\\
    p^N,&k=N.
  \end{cases} \label{eq:binomial-random-mass}
\end{equation}
By Lemma~\ref{lem:binomial-weight} for $0<k<N$, with $k=0,N$ immediate by definition,
\begin{equation}
  \E\widehat\pi_k(Z_2)
  =\binom Nk p^kq^{N-k}
  =\Pp\{\Bin(N,p)=k\}.
  \label{eq:binomial-unbiasedness}
\end{equation}

Define
\begin{align}
  \rho(u)
  &=\begin{cases}
    \dfrac{2}{\sqrt p}\{\sqrt{Nq}-\sqrt{N-u}\},&u\leq Np,\\[2mm]
    \dfrac{2}{\sqrt q}\{\sqrt u-\sqrt{Np}\},&u\geq Np,
  \end{cases} \label{eq:binomial-rho}\\
  \rho^{-1}(z)
  &=\begin{cases}
    N-(\sqrt{Nq}-\tfrac12\sqrt p\,z)^2,&z<0,\\[1mm]
    (\sqrt{Np}+\tfrac12\sqrt q\,z)^2,&z\geq0.
  \end{cases} \label{eq:binomial-rho-inverse}
\end{align}
Let
\begin{equation}
  s=\max\{0,\min\{N-4,\lfloor Np\rfloor-2\}\},
  \qquad t=\min\{N,s+4\},
  \qquad S=\{s,\ldots,t\}. \label{eq:binomial-repair-set}
\end{equation}
Define
\[
  T(z)=\begin{cases}
    \rho^{-1}(z)-1,&z<0,\\
    \rho^{-1}(z)+2,&z\geq0,
  \end{cases}
  \qquad Y=\lfloor T(Z_1)\rfloor.
\]
For $k\in \{0,\dots, N\}\setminus S$, set $I_k=T^{-1}([k,k+1))$:
\begin{equation}
  I_k=
  \begin{cases}
    [\rho(k+1),\rho(k+2)),&k+2\leq Np,\\[1mm]
    [\rho(k-2),\rho(k-1)),&k-2\geq Np.
  \end{cases} \label{eq:binomial-cells}
\end{equation}

Define
\[
  A=
  \begin{cases}
    \displaystyle
    \frac{\widehat\pi_Y(Z_2)}{|I_Y|\phi(Z_1)},
    &Y\in\{0,\ldots,N\}\setminus S,\\[2mm]
    0,&\text{otherwise}.
  \end{cases}
\]
\begin{lemma}[Binomial envelope]\label{lem:binomial-envelope}
  For $k\in\{0,\ldots,N\}\setminus S$ and $z\in\R$,
  \begin{equation}\label{eq:binomial-envelope}
    \widehat\pi_k(z)
    \leq |I_k|\inf_{w\in I_k}\phi(w).
  \end{equation}
\end{lemma}

\begin{proof}
  See Appendix~\ref{app:binomial-envelope}.
\end{proof}

\begin{algorithm}[H]
  \caption{$\Bin(N,p)$ sampler}
  \label{alg:binomial}
  \begin{algorithmic}[1]
    \Require integer $N\geq1$ and $0\leq p\leq1$
    \State \textbf{if} $p=0$ \textbf{then return} $0$; \textbf{if} $p=1$ \textbf{then return} $N$
    \State $q\gets1-p$, $s\gets\max\{0,\min\{N-4,\lfloor Np\rfloor-2\}\}$, $t\gets\min\{N,s+4\}$
    \State $U_1,U_2,U_3\overset{\mathrm{i.i.d.}}{\sim}\Unif(0,1)$
    \State $R\gets\sqrt{-2\log U_1}$, $Z_1\gets R\cos(2\pi U_2)$, $Z_2\gets R\sin(2\pi U_2)$
    \State $K\gets\lfloor\rho^{-1}(Z_1)-1\rfloor$ \textbf{if} $Z_1<0$
    \State \hspace{1.2em}\textbf{else} $K\gets\lfloor\rho^{-1}(Z_1)+2\rfloor$
    \State $A\gets0$
    \State \textbf{if} $0\leq K\leq N$ and $(K<s$ or $K>t)$ \textbf{then}
    \State \hspace{1.2em}$L\gets\frac{2}{\sqrt p}(\sqrt{N-K-1}+\sqrt{N-K-2})^{-1}$ \textbf{if} $K+2\leq Np$
    \State \hspace{2.4em}\textbf{else} $L\gets\frac{2}{\sqrt q}(\sqrt{K-1}+\sqrt{K-2})^{-1}$
    \State \hspace{1.2em}$A\gets\widehat\pi_K(Z_2)/(L\phi(Z_1))$
    \State \textbf{if} $A>0$ and $U_3<A$ \textbf{then return} $K$
    \State $V\gets(U_3-A)/(1-A)$, $w_0\gets1$
    \For{$j=0,1,2,3$}
      \State $w_{j+1}\gets w_j\dfrac{N-s-j}{s+j+1}\dfrac pq$ \textbf{if} $s+j<t$
      \State $w_{j+1}\gets0$ \textbf{otherwise}
    \EndFor
    \State \textbf{return} $s+\Cat(V;w_0,w_1,w_2,w_3,w_4)$
  \end{algorithmic}
\end{algorithm}

\begin{theorem}[Binomial]\label{thm:binomial}
  For every integer $N\geq1$ and $p\in[0,1]$, \Cref{alg:binomial} returns a
  $\Bin(N,p)$ variate.  Thus the $\Bin(N,p)$ law is 3-simple.
\end{theorem}

\begin{proof}
  The cases $p=0$ and $p=1$ are immediate.  For $0<p<1$, let
  $\pi_k=\Pp\{\Bin(N,p)=k\}$.  Apply
  Proposition~\ref{prop:finite-gaussian-repair}, where the transformation $T$
  satisfies the required hypotheses, $\widehat\pi_k(Z_2)$ is unbiased by
  \eqref{eq:binomial-unbiasedness}, and the inequality
  \eqref{eq:discrete-repair-pi-bound} is given by
  \eqref{eq:binomial-envelope}.  Finally,
  $w_j\propto\pi_{s+j}$, so the complement arm draw has the required distribution.
\end{proof}

\section{Conclusion}
We have shown that the gamma, Poisson, and binomial families admit extended one-liners.
As these distributions play a central role in probability theory, several additional families admit
extended one-liners by transformation or composition:
negative binomial, beta, inverse gamma, (noncentral) chi-squared and F, beta-binomial, Skellam,
multinomial, and Dirichlet (for fixed dimension).
Although all three distributions are 3-simple, the transformations presented here are complicated
and require computing several computationally expensive operations; finding simpler one-liners would
be of practical interest.
We refer the reader to recent work by Devroye \cite{Devroye2026} for simplifications,
extensions, and remaining open questions.
\label{sec:conclusion}

\section*{Acknowledgements}
The results in this paper stemmed from significant interaction with ChatGPT 5.5, 5.6 Sol, and Codex,
which found the constructions and proofs. Codex was additionally used to verify the results in Lean
and for help with typesetting; Astra was used for review.
The author independently verified the results and takes full responsibility for the paper and any
errors.
\appendix

\section{Gamma Extended One-Liner Below One}\label{app:gamma-below-one}

The following three-uniform gamma extended one-liner for $a<1$ is from
\cite{GreavesBelowOne}.
\begin{algorithm}[H]
  \caption{\textsc{GammaBelowOne}$(a)$: $\Gam(a)$ sampler for $0<a<1$}
  \label{alg:gamma-below-one}
  \begin{algorithmic}[1]
    \Require $0<a<1$
    \State $U_1,U_2,U_3\overset{\mathrm{i.i.d.}}{\sim}\Unif(0,1)$
    \State $b\gets1-a$, $\kappa\gets\sin(\pi a)/(\pi ab)$
    \State $R\gets(U_1/b)^{1/a}$ if $U_1\leq b$, else $R\gets\{(1-U_1)/a\}^{1/b}$
    \State $P\gets R/(1+R)$ if $U_1\leq b$, else $P\gets1/(1+R)$
    \State $A\gets P+\kappa(b-P)\max(P,1-P)$
    \State \textbf{return} $-\log U_3\cdot
    \begin{cases}
      PU_2/A,&U_2\leq A,\\
      P+(1-P)(U_2-A)/(1-A),&U_2>A.
    \end{cases}$
  \end{algorithmic}
\end{algorithm}

\section{Random Weight Proofs}\label{app:saddlepoint-weight-proofs}

\subsection{Reciprocal-Gamma Weight}\label{app:reciprocal-gamma-weight}

\begin{proof}[Proof of Lemma~\ref{lem:reciprocal-gamma-weight}]
  We use the identity
  \begin{equation}
    \frac{(e/a)^a}{\pi}\int_0^\pi e^{-a\Phi(\theta)}\dd\theta
    =\frac1{\Gamma(a+1)}, \label{eq:temme-identity}
  \end{equation}
  which follows from an application of the saddlepoint method to Hankel's integral for the
  reciprocal gamma function; see DLMF~\cite{DLMF}, Temme~\cite[Section 3.6.3]{Temme1996} for a
  proof. The density of $\Theta$ is
  \[
    \sqrt{\frac{2a}{\pi}}e^{-a\theta^2/2},
    \qquad \theta>0.
  \]
  Thus
  \[
    \E W_a(\Theta)
    =\frac{(e/a)^a}{\pi}\int_0^\pi e^{-a\Phi(\theta)}\dd\theta
    =\frac1{\Gamma(a+1)}.
  \]
  The upper bound follows from
  \[
    \frac{d}{d\theta}\left(\Phi(\theta)-\frac{\theta^2}{2}\right)
    =\frac{(\theta\cos\theta-\sin\theta)^2}{\theta\sin^2\theta}\geq0,
  \]
  and $\lim_{\theta\downarrow0}\{\Phi(\theta)-\theta^2/2\}=0$.
\end{proof}

\subsection{Binomial Weight}\label{app:binomial-weight}

\begin{proof}[Proof of Lemma~\ref{lem:binomial-weight}]
  We use an integral due to Mallows:
  \[
    \frac1\pi\int_0^\pi
    \frac{(\sin\theta)^t}
    {(\sin(a\theta))^{at}(\sin(b\theta))^{bt}}\dd\theta
    =\frac{\Gamma(t+1)}{\Gamma(at+1)\Gamma(bt+1)},
    \qquad t\geq0,\quad a,b>0,\quad a+b=1.
  \]
  Mallows showed this probabilistically using negative moments of Kanter's representation of a
  positive stable random variable \cite{Williams1977,Kanter1975,Mallows1980};
  a direct proof is given by Evans, Ismail, and Stanton
  \cite[Eq.~(1.2)]{EvansIsmailStanton1982}.  Taking
  \[
    t=N,\qquad a=x=\frac{k}{N},
    \qquad b=y=1-\frac{k}{N},
  \]
  and noting that
  \[
    e^{-N\Phi_x(\theta)}
    =x^ky^{N-k}
    \frac{(\sin\theta)^N}
    {(\sin(x\theta))^k(\sin(y\theta))^{N-k}},
  \]
  gives
  \begin{equation}
    \frac1\pi\int_0^\pi e^{-N\Phi_x(\theta)}\dd\theta
    =\binom Nkx^ky^{N-k} \label{eq:finite-binomial-contour}
  \end{equation}
  The density of $\Theta$ is
  \[
    \sqrt{\frac{2Nxy}{\pi}}e^{-Nxy\theta^2/2},
    \qquad \theta>0.
  \]
  Thus
  \[
    \E W_{N,k}(\Theta)
    =\frac1\pi\int_0^\pi e^{-N\Phi_x(\theta)}\dd\theta
    =\binom Nkx^ky^{N-k}.
  \]

  For the upper bound, by Euler's sine product formula, for
  $0\leq\theta<\pi$,
  \[
    \log\frac{\theta}{\sin\theta}
    =\sum_{n\geq1}\sum_{m\geq1}
    \frac1m\left(\frac{\theta^2}{n^2\pi^2}\right)^m
    =\sum_{m\geq1}c_m\theta^{2m},
    \qquad
    c_m=\frac{\zeta(2m)}{m\pi^{2m}}>0,
    \quad c_1=\frac16.
  \]
  Therefore
  \[
    \Phi_x(\theta)
    =\sum_{m\geq1}c_m\theta^{2m}
    \{1-x^{2m+1}-y^{2m+1}\}.
  \]
  Since $x+y=1$, the $m=1$ term is
  $c_1\theta^2(1-x^3-y^3)=xy\theta^2/2$, and every remaining term is
  nonnegative.  Hence $\Phi_x(\theta)-xy\theta^2/2\geq0$.
\end{proof}

\section{Envelope Proofs}\label{app:envelope-proofs}

\subsection{Gamma Envelope}\label{app:gamma-envelope}

Set
\[
  M_a=\frac{a(e/a)^a}{\sqrt{2\pi a}}.
\]
By Lemma~\ref{lem:reciprocal-gamma-weight},
\[
  0\leq\widehat w_a(z)\leq M_a.
\]
Define
\[
  L_a(x)=\log\frac{M_ag_a(x)}{\phi(x)}
  =\delta_a-J_a(x),
\]
where, for $x>-1/c$,
\begin{equation}
  J_a(x)=3d\left\{\frac{(cx)^3}{3}-\frac{(cx)^2}{2}
  +cx-\log(1+cx)\right\}
  =\frac{c^2x^4}{3}\int_0^1\frac{s^3}{1+cxs}\dd s. \label{eq:cubic-J}
\end{equation}

\begin{lemma}\label{lem:cubic}
  For $a\geq1$,
  \[
    0<\delta_a\leq\frac{2}{15},\qquad
    J_a(x)\geq\delta_a\min\{x^4/16,1\},
    \qquad J_a(x)<\frac43\quad(|x|\leq2).
  \]
\end{lemma}

\begin{proof}
  Put $r=1/(3d)=3c^2$.  Since $a\geq1$, we have $0<r\leq1/2$ and
  \[
    \delta_a=\frac13-\frac{2-r}{6r}\log(1+r).
  \]
  The bounds
  \[
    \frac{2z}{2+z}\leq\log(1+z)\leq z-\frac{z^2}{2(1+z)},
    \qquad z\geq0,
  \]
  give
  \[
    \delta_a\geq\frac{r(4+r)}{12(1+r)}>0,
    \qquad
    \delta_a\leq\frac{2r}{3(2+r)}
    =\frac{2c^2}{2+3c^2}\leq\frac{2}{15}.
  \]
  For $0\leq x<1/c$, we have $1-cxs\leq1+cxs$, so the integral in
  \eqref{eq:cubic-J} gives $J_a(-x)\geq J_a(x)$.  For $x>0$,
  differentiation under the integral and differentiation of
  \eqref{eq:cubic-J} give
  \[
    \frac{\dd}{\dd x}\frac{J_a(x)}{x^4}
    =-\frac{c^3}{3}\int_0^1\frac{s^4}{(1+cxs)^2}\dd s<0,
    \qquad
    J_a'(x)=\frac{3dc(cx)^3}{1+cx}>0.
  \]
  Thus $J_a(x)/x^4$ decreases and $J_a(x)$ increases on $(0,\infty)$.
  If $S$ is a random variable with density $4s^3$ on $[0,1]$ and
  mean $4/5$, then Jensen's inequality yields
  \[
    \int_0^1\frac{s^3}{1+2cs}\dd s
    \geq\frac{5}{4(5+8c)}.
  \]
  Thus,
  \[
    J_a(2)\geq\frac{20c^2}{15+24c}
    \geq\frac{2c^2}{2+3c^2}\geq\delta_a,
  \]
  where the middle inequality is equivalent to
  $30c^2-24c+5>0$.  It follows from the
  monotonicity properties above that
  \[
    J_a(x)\geq\delta_a\min\{x^4/16,1\}.
  \]
  Finally, for $|x|\leq2$, \eqref{eq:cubic-J} gives
  \[
    J_a(x)\leq\frac{c^2x^4}{12(1-2c)}
    \leq\frac{4c^2}{3(1-2c)}<\frac43,
  \]
  where the last inequality follows from
  $c\leq1/\sqrt6<\sqrt2-1$.
\end{proof}

\begin{proof}[Proof of Lemma~\ref{lem:gamma-envelope}]
  If $x\leq-1/c$, then $g_a(x)=r_a(x)=0$ because $-1/c<-2$.
  Assume $x>-1/c$ and put $u=|x|/2$.  On $u\leq1$,
  Lemma~\ref{lem:cubic} gives $L_a(x)\leq\delta_a(1-u^4)$.  By convexity,
  \[
    e^{\delta t}-1\leq t(e^\delta-1),
    \qquad 0\leq t\leq1,
  \]
  and
  \[
    1-u^4\leq\frac65e^{2u^2}(1-u),\qquad0\leq u\leq1,
  \]
  because $1-u^4=(1-u)(1+u)(1+u^2)$,
  $1+u^2\leq e^{u^2}$, and
  $1+u\leq(6/5)e^{u^2}$.  Therefore
  \[
    e^{L_a(x)}-1
    \leq(e^{\delta_a}-1)(1-u^4)
    \leq\frac65(e^{\delta_a}-1)e^{2u^2}(1-u)
    =\frac{M_ar_a(x)}{\phi(x)}.
  \]
  Furthermore, we have
  \[
    \frac{r_a(x)}{g_a(x)}
    =\frac65\eta_a(1-u)e^{J_a(x)+2u^2}.
  \]
  Since $\eta_a\leq\delta_a\leq2/15$ and $J_a(x)<4/3$, this gives
  \[
    \frac{r_a(x)}{g_a(x)}
    \leq\frac65\frac{2}{15}e^{4/3}
    \max_{0\leq u\leq1}(1-u)e^{2u^2}
    =\frac4{25}e^{4/3}<1.
  \]
  The maximum is attained at $u=0$, since the derivative is nonpositive on
  $[0,1]$.  On $u\geq1$, $r_a=0$ and $L_a\leq0$.
  Using $\widehat w_a\leq M_a$ gives the desired inequality.
\end{proof}

\subsection{Poisson Envelope}\label{app:poisson-envelope}

Fix $\lambda>0$.  Let $\rho$, $S$, and $I_k$ be as defined in
Section~\ref{sec:poisson}, set
\[
  D(t)=D_{\mathrm P}(t\Vert\lambda),
\]
and put
\[
  M(k)=
  \begin{cases}
    \displaystyle\frac{e^{-D(k)}}{\sqrt{2\pi k}},&k\geq1,\\[2mm]
    e^{-D(0)},&k=0.
  \end{cases}
\]
Write
\[
  d(x,y)=x\log(x/y)-x+y.
\]
Holding $y$ fixed, differentiating the difference of the left and right-hand side below shows
\begin{equation}\label{eq:poisson-divergence-bound}
  d(x,y)\geq
  \begin{cases}
    (x-y)^2/(2y),&0\leq x\leq y,\\
    2(\sqrt x-\sqrt y)^2,&x\geq y.
  \end{cases}
\end{equation}
These bounds give
\begin{equation}\label{eq:poisson-kl}
  D(t)=d(t,\lambda)\geq\frac12\rho(t)^2.
\end{equation}
For $j\geq1$, let
\[
  A_j=\frac{\sqrt{j(j+2)}}{j+1}
  \left(\frac{j}{j+1}\right)^j.
\]
The continuous function
\[
  A(t)=\frac{\sqrt{t(t+2)}}{t+1}
  \left(\frac{t}{t+1}\right)^t,
  \qquad t\geq1.
\]
is strictly decreasing, as
\[
  \frac{\dd}{\dd t}\log A(t)
  =\frac{t+1-t(t+2)\log(1+1/t)}{t(t+2)}<0,
\]
where the inequality follows from
$\log(1+1/t)>2/(2t+1)$.  Since
$A(t)\to e^{-1}$, we have
\begin{equation}\label{eq:A-bound}
  A_j>e^{-1}.
\end{equation}

\begin{lemma}[Poisson bound]\label{lem:poisson-cell}
  For every $k\in\mathbb N\setminus S$,
  \[
    M(k)\leq |I_k|\inf_{z\in I_k}\phi(z).
  \]
\end{lemma}

\begin{proof}
  Suppose first that $1\leq k$ and $k+2\leq\lambda$.  Then $I_k\subset (-\infty,0]$, hence
  \[
    |I_k|=\frac1{\sqrt\lambda},\qquad
    \inf_{z\in I_k}\phi(z)=\phi(\rho(k+1)).
  \]
  By \eqref{eq:poisson-kl} at $k+1$,
  \[
    \phi(\rho(k+1))
    \geq\frac{e^{-D(k+1)}}{\sqrt{2\pi}}.
  \]
  Therefore
  \[
    \begin{aligned}
      \frac{|I_k|\inf_{z\in I_k}\phi(z)}{M(k)}
      &\geq\sqrt{\frac{k}{\lambda}}
      e^{D(k)-D(k+1)}\\
      &=e\sqrt{k\lambda}\,\frac{k^k}{(k+1)^{k+1}}
      \geq eA_k>1
    \end{aligned}
  \]
  by \eqref{eq:A-bound}.  Similarly, if $k=0$, then $\lambda\geq2$, and
  \[
    \frac{|I_0|\inf_{z\in I_0}\phi(z)}{M(0)}
    \geq\frac{e^{D(0)-D(1)}}{\sqrt{2\pi\lambda}}
    =e\sqrt{\frac{\lambda}{2\pi}}
    \geq \frac{e}{\sqrt\pi}>1.
  \]
  For the remaining case $k-1\geq\lambda$,
  \[
    |I_k|=2(\sqrt k-\sqrt{k-1})>\frac1{\sqrt k}.
  \]
  Moreover, by \eqref{eq:poisson-kl}, \[
    \inf_{z\in I_k}\phi(z)
    \geq \phi(\rho(k))
    \geq \frac{e^{-D(k)}}{\sqrt{2\pi}},
  \]
  and therefore
  \[
    |I_k|\inf_{z\in I_k}\phi(z)
    >\frac{e^{-D(k)}}{\sqrt{2\pi k}}
    =M(k).
  \]
\end{proof}

\begin{proof}[Proof of Lemma~\ref{lem:poisson-envelope}]
  For $k\geq1$, the upper bound in Lemma~\ref{lem:reciprocal-gamma-weight} gives
  \[
    \begin{aligned}
      \widehat\pi_k(z)
      &=e^{-\lambda}\lambda^kW_k(|z|/\sqrt k)\\
      &\leq e^{-\lambda}\lambda^k\frac{(e/k)^k}{\sqrt{2\pi k}}
      =\frac{e^{-D(k)}}{\sqrt{2\pi k}}
      =M(k).
    \end{aligned}
  \]
  If $k=0$, then $\widehat\pi_0(z)=e^{-\lambda}=M(0)$.  Hence, for
  $k\notin S$, Lemma~\ref{lem:poisson-cell} gives
  \[
    \widehat\pi_k(z)\leq M(k)
    \leq |I_k|\inf_{w\in I_k}\phi(w).\qedhere
  \]
\end{proof}

\subsection{Binomial Envelope}\label{app:binomial-envelope}

The case $N=1$ is vacuously true because $S=\{0,1\}$.  Fix
$N\geq2$ and $p\in(0,1)$, and write $q=1-p$ and $\mu=Np$.  Let $\rho$,
$S$, and $I_k$ be as defined in Section~\ref{sec:binomial}, and set
\[
  D(t)=N D_{\mathrm B}(t/N\Vert p),
  \qquad 0\leq t\leq N,
\]
and put
\[
  M(k)=
  \begin{cases}
    \displaystyle
    \frac{e^{-D(k)}}
    {\sqrt{2\pi k(N-k)/N}},&1\leq k\leq N-1,\\[3mm]
    e^{-D(k)},&k\in\{0,N\}.
  \end{cases}
\]
For $j\geq1$, let
\[
  B_j=\sqrt{\frac{j+2}{j}}
  \left(\frac{j+2}{j+1}\right)^{j+2},
  \qquad
  \widetilde B_j=\sqrt{\frac{j+1}{j}}\,B_j.
\]
Since $\sqrt{(j+2)/j}>1$,
\begin{equation}\label{eq:B-bound}
  B_j>
  \left(1+\frac1{j+1}\right)^{j+2}>e.
\end{equation}

\begin{lemma}[Binomial bound]\label{lem:binomial-cell}
  For every $k\in\{0,\ldots,N\}\setminus S$,
  \[
    M(k)
    \leq |I_k|\inf_{z\in I_k}\phi(z).
  \]
\end{lemma}

\begin{proof}
  Since
  \[
    D(t)=d(t,\mu)+d(N-t,N-\mu),
  \]
  \eqref{eq:poisson-divergence-bound} yields
  \begin{equation}\label{eq:binomial-kl}
    D(t)\geq\frac12\rho(t)^2.
  \end{equation}
  Indeed, when $t\leq\mu$,
  \eqref{eq:poisson-divergence-bound} gives
  \[
    D(t)
    \geq (\mu-t)^2/(2\mu)
    +2\bigl(\sqrt{N-t}-\sqrt{N-\mu}\bigr)^2
    \geq \frac{2N}{\mu}
    \bigl(\sqrt{N-t}-\sqrt{N-\mu}\bigr)^2,
  \]
  where the last inequality follows from
  \[
    (a+b)^2+4\mu-4N=(a-b)(a+3b)\geq0,
  \]
  with $a=\sqrt{N-t}$ and $b=\sqrt{N-\mu}$.
  The case $t\geq\mu$ follows by symmetry with successes and failures
  interchanged: $p\mapsto 1-p$, $k\mapsto N-k$.

  Suppose first that $1\leq k$ and $k+2\leq\mu$, and put
  $j=N-k-2\geq1$.  Then $I_k\subset (-\infty,0]$, hence
  \[
    \inf_{z\in I_k}\phi(z)=\phi(\rho(k+1)).
  \]
  Moreover,
  \[
    \begin{aligned}
      |I_k|
      &=2\sqrt{\frac N\mu}\left(\sqrt{j+1}-\sqrt j\right)\\
      &=\frac{2\sqrt{N/\mu}}{\sqrt{j+1}+\sqrt j}
      >\sqrt{\frac{N}{\mu(j+1)}}.
    \end{aligned}
  \]
  Using \eqref{eq:binomial-kl} at $k+1$ gives
  \[
    \phi(\rho(k+1))
    \geq\frac{e^{-D(k+1)}}{\sqrt{2\pi}},
  \]
  therefore
  \[
    \frac{|I_k|\inf_{z\in I_k}\phi(z)}{M(k)}
    \geq
    A_k\widetilde B_j\,
    \frac{j}{N-\mu}\sqrt{\frac{\mu}{k+2}}.
  \]
  As $j\geq N-\mu$ and $\mu\geq k+2$, the last two factors are at
  least one.  Hence, by \eqref{eq:A-bound} and \eqref{eq:B-bound}, the ratio
  is greater than $A_kB_j>1$.  Similarly, if $k=0$ and $k\not\in S$, then
  $\mu\geq2$, and
  \[
    \frac{|I_0|\inf_{z\in I_0}\phi(z)}{M(0)}
    \geq
    \left(\frac{N}{N-1}\right)^{N-1}
    \frac{N}{N-\mu}
    \sqrt{\frac{\mu N}{2\pi(N-1)}}
    >\frac{2}{\sqrt\pi}>1.
  \]
  The case $\mu+2\leq k\leq N$ follows by symmetry, interchanging successes and failures.
\end{proof}

\begin{proof}[Proof of Lemma~\ref{lem:binomial-envelope}]
  For $1\leq k\leq N-1$, the upper bound in
  Lemma~\ref{lem:binomial-weight} gives
  \[
    \widehat\pi_k(z)\leq M(k).
  \]
  If $k\in\{0,N\}$, then $\widehat\pi_k(z)=M(k)$.  Hence, for $k\notin S$,
  Lemma~\ref{lem:binomial-cell} gives
  \[
    \widehat\pi_k(z)\leq M(k)
    \leq |I_k|\inf_{w\in I_k}\phi(w).\qedhere
  \]
\end{proof}

\end{document}